%% file: root.tex
\documentclass[letterpaper, 10 pt, conference]{ieeeconf}  % Comment this line out if you need a4paper

\IEEEoverridecommandlockouts                              % This command is only needed if 
\input{notation}

\title{\LARGE \bf Guarding a Target Area from a Heterogeneous Group of Cooperative Attackers}

\author{Yoonjae Lee \and Goutam Das \and Daigo Shishika \and Efstathios Bakolas % <-this % stops a space
\thanks{Y. Lee and E. Bakolas are with the Department of Aerospace Engineering
and Engineering Mechanics, The University of Texas at Austin,
Austin, Texas 78712-1221, USA. Emails: \texttt{yol033@utexas.edu; bakolas@austin.utexas.edu}}%
\thanks{G. Das and D. Shishika are with the College of Engineering and Computing, George Mason
University, Fairfax, VA 22033, USA. Emails: \texttt{gdas@gmu.edu; dshishik@gmu.edu}}
}

\begin{document}

\maketitle
\thispagestyle{empty}
\pagestyle{empty}

%%%%%%%%%%%%%%%%%%%%%%%%%%%%%%%%%%%%%%%%%%%%%%%%%%%%%%%%%%%%%%%%%%%%%%%%%%%%%%%%
\begin{abstract} 
In this paper, we investigate a multi-agent target guarding problem in which a single defender seeks to capture multiple attackers aiming to reach a high-value target area. In contrast to previous studies, the attackers herein are assumed to be \emph{heterogeneous} in the sense that they have not only different speeds but also different weights representing their respective degrees of importance (e.g., the amount of allocated resources). The objective of the attacker team is to jointly minimize the weighted sum of their final levels of proximity to the target area, whereas the defender aims to maximize the same value. Using geometric arguments, we construct candidate equilibrium control policies that require the solution of a (possibly nonconvex) optimization problem. Subsequently, we validate the optimality of the candidate control policies using parametric optimization techniques. Lastly, we provide numerical examples to illustrate how cooperative behaviors emerge within the attacker team due to their heterogeneity.
\end{abstract}

%%%%%%%%%%%%%%%%%%%%%%%%%%%%%%%%%%%%%%%%%%%%%%%%%%%%%%%%%%%%%%%%%%%%%%%%%%%%%%%%
\section{Introduction}

Target guarding games represent a challenging category of dynamic optimization problems that involve multiple mobile decision makers with conflicting objectives. Typically, the decision makers are classified as attackers and defenders. Defenders seek to capture attackers who attempt to reach a high-value target area while avoiding or delaying capture. These (reach-avoid) problems have gained significant attention in recent years within the fields of controls, robotics, and various engineering domains. The heightened interest can be attributed to the growing demand of autonomy in both civilian and military applications, such as the Counter-Unmanned Aircraft Systems (C-UAS) technologies \cite{cuas}.

A simple version of a target guarding game has been first introduced and solved using geometric methods by Issacs in his seminal work \cite{isaacs1965differential}. In recent years, numerous variations of this game have been explored, ranging from area/perimeter defense \cite{adler,shishika2020cooperative,9743548,VonMoll2020BD,yan_matching,garcia_multiple,lee_cdc,lee2021guarding,moll_circ,FU2023110811,lee_2024,DOROTHY2024111587} and reconnaissance \cite{liang_recon,lee_recon} to the defense of moving \cite{Das2022,deng_flow} and maneuvering \cite{LIANG201958,8340791} targets, each offering unique insights and solutions. For a recent survey on this topic, we refer the reader to \cite{10.3389/fcteg.2022.1093186}.

In the majority of cases, the number of defenders either equals or exceeds the number of attackers. However, there is relatively limited literature on scenarios where attackers outnumber defenders and where one defender can capture multiple attackers in a sequential manner. A few articles devoted to such a class of games are as follows. Scenarios with a single defender protecting a half-planar target area from two cooperative attackers have been examined in \cite{9029340,yan2021cooperative,deng_sing}. A more general scenario involving an arbitrary number of noncooperative attackers has been studied in \cite{fu2021strategies}. Solution methods for optimizing capture order have been developed in \cite{zepp2022autonomous}. Algorithmic perspectives on a similar class of games have been explored in \cite{pourghorban2022target} and \cite{bajaj_tro}.

In sharp contrast to the key references above, this paper focuses on scenarios where an arbitrary number of attackers characterized by \emph{heterogeneous} speed ratios and degrees of importance collaborate to minimize the weighted sum of their final levels of proximity to a convex target area. This problem appears to be extremely challenging since the equilibrium control policies of the agents can only be obtained through, despite the convexity of the target area, the globally optimal solution of a possibly nonconvex optimization problem.

The main contributions of this paper can be summarized as follows: (i) Generalization of solutions developed in prior works \cite{9029340,yan2021cooperative,fu2021strategies} to scenarios where the number of attackers and the dimension and shape of the target area can be arbitrary; (ii) Consideration of heterogeneous characteristics (specifically, varying speed ratios and resource-related weights) within the attacker team; (iii) Rigorous validation of geometrically constructed equilibrium control policies and the associated nonconvex optimization problem for finding optimal capture points; and (iv) Provision of nontrivial numerical examples demonstrating how cooperative behaviors emerge within the attacker team as a result of their heterogeneity.

The remainder of the paper is structured as follows. In Section \ref{sec:probform}, a target guarding game involving a team of cooperative and heterogeneous attackers is formulated. In Section \ref{sec:1v1}, a solution to the special case with only one attacker is briefly introduced. In Section \ref{sec:main}, the main results of the paper are presented. In Section \ref{sec:sim}, illustrative numerical simulations are offered.
Finally, concluding remarks are provided in Section \ref{sec:concl}.

\section{Problem Formulation} \label{sec:probform}
% ---- Preliminaries

\subsection{Notation}

In this paper, we denote by $\|\cdot\|$ the two norm of a vector and by $[n]$ the set $\{1,\dots,n\}$, where $n$ is a positive integer.

\subsection{Problem Statement} \label{sec:game}

Consider a game involving a single defender ($D$) and $m$ attackers ($A_1,\dots,A_m$) interacting in $\bR^n$, as illustrated in Figure \ref{fig:illustration}. Let us denote the attacker team as $A$. The target area that $A$ and $D$ attempt to reach and guard, respectively, is given as
\begin{align}
    \cT = \lcb x \in \bR^n : h(x) \leq 0 \rcb,
\end{align}
where $h : \bR^n \rightarrow \bR$ is a twice continuously differentiable and convex function whose value indicates the level of proximity of a point $x \in \bR^n \backslash \cT$ to $\cT$ (e.g., $h(x) = \|x\|^2 - r^2$ if $\cT$ is a Euclidean ball with some radius $r \geq 0$). Suppose each agent has simple-motion kinematics:
\begin{subequations} \label{eq:kinematics}
    \begin{align}
        \dot x_{A_i} &= \nu_i u_i, & x_{A_i}(t_0) &= x_{A_i}^0, & i \in [m],
        \\
        \dot x_D &= v, & x_D(t_0) &= x_{D}^0,
    \end{align}
\end{subequations}
where $x_{A_i},x_D \in \bR^n$, $x_{A_i}^0,x_{D}^0 \in \bR^n$, and $u_i,v \in \cU = \{ \xi \in \bR^n : \|\xi\| \leq 1 \}$ denote the position, initial position, and control input of $A_i$ and $D$, respectively, and $\nu_i \in (0,1)$ denotes the speed ratio of $A_i$ and $D$. Note that $D$ moves with unitary speed and is faster than all attackers.

Feedback (perfect state) information pattern is assumed, where both $A$ and $D$ can perfectly measure the global position (or state) $\mbx = [x_{A_1}^\top,\dots,x_{A_m}^\top,x_D^\top]^\top$ at all times and adopt (piecewise-continuous) feedback strategies $\gamma_{A} = (\gamma_{A_1},\dots,\gamma_{A_m}) : \bR^{n(m+1)} \rightarrow \cU^m$ and $\gamma_D : \bR^{n(m+1)} \rightarrow \cU$. For notational brevity, let $\mbx_{i:j} = [x_{A_i}^\top,\dots,x_{A_j}^\top,x_D^\top]^\top$ for any $i,j \in [m]$ with $i \leq j$ (when $i=j$, we simply write $\mbx_i$).

Under the assumption that the indexing of $A$ corresponds to the order in which they will be captured, the performance index of the game, which is to be jointly minimized by $A$ and solely maximized by $D$, is defined to be
\begin{align} \label{eq:payoff}
    J(\gamma_A,\gamma_D;\mbx^0) = \sum\nolimits_{i=1}^m \theta_i h(x_{A_i}(t_i)),
\end{align}
where $\mbx^0$ is the initial state of the game, $t_i = \inf \{t \geq t_{i-1} : \mbx(t) \in \cC_i \}$ is the capture time for $A_i$ with
\begin{align} \label{eq:terset}
    \cC_i = \lcb \mbx : \|x_{A_i}-x_D\| = 0 ~\textrm{and}~ h(x_{A_i}) > 0 \rcb,
\end{align}
and $\theta_1, \dots, \theta_m$ are scalar parameters representing the attacker's respective degrees of importance (e.g., the amount of allocated resources), which satisfy $\theta_i > 0$ (for all $i \in [m]$) and $\sum_{i=1}^m \theta_i = 1$.

\begin{figure}
    \centering
    \includegraphics[scale=0.8]{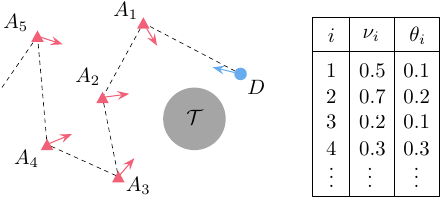}
    \caption{Game of protecting a target area ($\cT$) from attackers ($A_i$) with heterogeneous speed ratios ($\nu_i$) and weights ($\theta_i$). The dashed lines depict the order of capture.}
    \label{fig:illustration}
\end{figure}

The main problem of the paper is stated below.

\begin{problem} \label{prob:main}
    For the game defined by \eqref{eq:kinematics}, \eqref{eq:payoff}, and \eqref{eq:terset} (for all $i \in [m]$), find the capturable set $\cX \subset \bR^{n(m+1)}$ (i.e., the set of initial states in which $D$ is guaranteed to capture all $m$ attackers outside $\cT$) and the corresponding equilibrium strategy pair $(\gamma_A^\star,\gamma_D^\star)$ satisfying the saddle-point condition
    \begin{align} \label{eq:sadpoint}
        J(\gamma_A^\star,\gamma_D;\mbx^0) \leq J(\gamma_A^\star,\gamma_D^\star;\mbx^0) \leq J(\gamma_A,\gamma_D^\star;\mbx^0),
    \end{align}
    for all admissible $\gamma_A$ and $\gamma_D$, and all $\mbx^0 \in \cX$.
\end{problem}

\section{Preliminary: Single-Attacker Game} \label{sec:1v1}

We briefly review the solution of the single-attacker game (i.e., $m=1$), which has been extensively studied in the literature (see, e.g., \cite{FU2023110811,lee_2024}).

Suppose at $t = t_{m-1}$, all attackers but $A_m$ have been captured. The performance index of the game reduces to
\begin{align} \label{eq:one_payoff}
    J^{(m)} = \Phi^{(m)}(\mbx_m(t_m)),
\end{align}
where $\Phi^{(m)}(\mbx_m) = \theta_m h(x_{A_m})$ is a terminal cost function. Note that the weight $\theta_m$ can be omitted without loss of generality, yet we retain it for consistency with the results to be presented in the following sections.

It follows readily from the first-order necessary conditions for saddle-point equilibria \cite[Theorem~8.2]{bacsar1998dynamic} that the equilibrium trajectory of this game over the regular region (where the Value of the game is continuously differentiable) is a straight line. From this geometrical property, one can deduce that the optimal terminal position (i.e., capture point) of $A_m$ can be found by solving the (convex) optimization problem
\begin{equation} \label{eq:one_main}
    \begin{aligned}
        \min_{x} \quad& \theta_m h(x)
        \\
        \mathrm{s.t.} \quad& \frac{\|x-x_{A_m}\|}{\nu_m} - \|x-x_D\| \leq 0,
    \end{aligned}
\end{equation}
whose feasible set
\begin{align}
    \cA^{(m)}(\mbx_m) = \lcb x : \frac{\|x-x_{A_m}\|}{\nu_m} - \|x-x_D\| \leq 0 \rcb, \nonumber
\end{align}
is referred to as the safe-reachable region of $A_m$.

A key preliminary result (which will be extended later in this paper) is that the (global) optimal value of problem \eqref{eq:one_main}, which can be expressed by
\begin{align}
    V^{(m)}(\mbx_m) = \min_{x} \lcb \theta_m h(x) : x \in \cA^{(m)}(\mbx_m) \rcb,
\end{align}
satisfies the related Isaacs equation \cite{isaacs1965differential} in the classical sense over the capturable set
\begin{align}
    \cX^{(m)} &= \left\{ \mbx_m : V^{(m)}(\mbx_m) > 0 \right\}.
\end{align}
Additionally, the equilibrium strategies establishing a saddle point of the game over $\cX^{(m)}$ can be derived to be
\begin{subequations} \label{eq:eqstr_one}
    \begin{align}
        \gamma_{A_m}^\star(\mbx_m) &= \frac{p_m^\star(\mbx_m)-x_{A_m}}{\|p_m^\star(\mbx_m)-x_{A_m}\|},
        \\
        \gamma_D^\star(\mbx_m) &= \frac{p_m^\star(\mbx_m)-x_D}{\|p_m^\star(\mbx_m)-x_D\|},
    \end{align}
\end{subequations}
where $p_m^\star(\mbx_m)$ denotes the (unique) solution of \eqref{eq:one_main}. In words, if $\mbx_m(t) \in \cX^{(m)}$, then the optimal strategies for both $A_m$ and $D$ are to align their control inputs (i.e., directions) with their respective lines-of-sight to the current optimal capture point $p_m^\star(\mbx_m(t))$.

\section{Main Results} \label{sec:main}

In this section, we present the main results of the paper, whereby we extend the preliminary results provided in the previous section to multi-attacker scenarios.

\subsection{Optimization Formulation of a Geometric Solution}

We begin by extending the straight-line property of the equilibrium trajectory observed in the previous section to the more general case involving a group of attackers with varying speed ratios and weights.

Let us decompose the game described in Problem \ref{prob:main} into a series of $m$ phases, where each $i$-th phase can be seen as a game involving the agents $A_i,\dots,A_m,D$ and occurring within the time interval $[t_i,t_{i+1}]$. The performance indices of all $m$ phases can be recursively defined in a backward manner as follows:
\begin{align}
    J^{(i)} = \Phi^{(i)}(\mbx_{i:m}(t_i)), \qquad \fa i \in [m],
\end{align}
with $\Phi^{(m)}$ as defined below \eqref{eq:one_payoff} and, for all $i \in [m-1]$,
\begin{align} \label{eq:tercost}
    \Phi^{(i)}(\mbx_{i:m}) = \theta_{i} h(x_{A_i}) + V^{(i+1)}(\mbx_{i+1:m}),
\end{align}
where $V^{(i+1)}$ is the Value of the $(i+1)$-th phase. Since all cost functions $\Phi^{(i)}$ are of terminal type, it follows again from \cite[Theorem~8.2]{bacsar1998dynamic} that the equilibrium trajectory in each phase is a straight line.

Let $p_i$ denote the capture point of $A_i$, and let $\mbp_{i:j} = [p_i^\top,\dots,p_j^\top]^\top$ for any $i,j \in [m]$ with $i \leq j$ (when $i=1$ and $j=m$, we simply write $\mbp$). In consideration of the above geometrical property, the optimization problem that each $A_i$ needs to solve individually to find its optimal capture point $p_i$, given information about the decisions made by all preceding attackers (i.e., $p_1,\dots,p_{i-1}$), can be formulated as the following hierarchical optimization problem:
\begin{equation} \label{eq:prob_series}
    \begin{aligned}
        \min_{p_i} \quad& \theta_i h(p_i) + \tilde V^{(i+1)}(p_i)
        \\
        \mathrm{s.t.} \quad& \frac{\|p_i-x_{A_i}\|}{\nu_i} - \sum\nolimits_{j=1}^{i} \|p_j-p_{j-1}\| \leq 0,
    \end{aligned}
\end{equation}
where $\tilde V^{(i+1)}(p_i)$ denotes the optimal value of the above problem with the index $i$ replaced by $i+1$ (which is a function of $p_i$). Note that $p_0 = x_D$. The constraint function here measures the difference between the minimum times needed for $A_i$ and $D$ to reach a point $p_i \in \bR^n$, provided that $D$ has to traverse waypoints $p_1,\dots,p_{i-1}$ before pursuing $A_i$. Hence the feasible set in \eqref{eq:prob_series} can be seen as the time-extended safe-reachable region of $A_i$ (see, e.g., \cite{yan2021cooperative,fu2021strategies} for similar definitions).

Albeit seemingly complex, solutions to the series of these hierarchical optimization problems at phase $i$ can be obtained at once by solving the single-level problem
\begin{equation} \label{eq:main_prob}
    \begin{aligned}
        \min_{\mbp_{i:m}} \quad& f_\theta^{(i)}(\mbp_{i:m}) = \sum\nolimits_{j=i}^m \theta_j h(p_j)
        \\
        \mathrm{s.t.} \quad& g_j^{(i)}(\mbp_{i:m},\mbx_{i:m}) \leq 0, \quad j \in [m] \backslash [i-1],
    \end{aligned} \tag{P}
\end{equation}
where
\begin{align} \label{eq:g_i}
    g_j^{(i)}(\mbp_{i:m},\mbx_{i:m}) = \frac{\|p_j-x_{A_j}\|}{\nu_j} - \sum_{k=i}^j \|p_k-p_{k-1}\|.
\end{align}
Here, we slightly abuse the notation by letting $p_{i-1} = x_D$. In view of the straight-line property, candidate equilibrium strategies for an $i$-th phase can be constructed as
\begin{subequations} \label{eq:eqstr_multi}
    \begin{align} 
        \gamma_{A_j}^{(i) \star}(\mbx_{i:m}) &= \frac{p_{j}^\star(\mbx_{i:m}) - x_{A_i}}{\|p_{j}^\star(\mbx_{i:m}) - x_{A_i}\|}, ~ j \in [m] \backslash [i-1],
        \\
        \gamma_D^{(i) \star}(\mbx_{i:m}) &= \frac{p_i^\star(\mbx_{i:m}) - x_D}{\|p_i^\star(\mbx_{i:m}) - x_D\|},
    \end{align}
\end{subequations}
where $\mbp_{i:m}^\star = [p_i^\star(\mbx_{i:m})^\top,\dots,p_m^\star(\mbx_{i:m})^\top]^\top$ is a global minimizer of problem \eqref{eq:main_prob}.

\begin{remark} \label{rem:first}
    As to be shown in Theorem \ref{theo:multi}, solutions for \eqref{eq:main_prob} will remain unchanged in all subsequent phases along the equilibrium trajectory. Thus, we will henceforth focus mainly on the initial phase of the game (i.e., $i=1$).
\end{remark}

In accordance with Remark \ref{rem:first}, in the rest of this section, we will fix $i=1$ and suppress any superscripts indicating the phase index (e.g., $\Phi(\cdot) = \Phi^{(1)}(\cdot)$).

For the forthcoming analysis, a few definitions are introduced. We define the feasible set of \eqref{eq:main_prob} as
\begin{align}
    \cA(\mbx) = \lcb \mbp \in \bR^{mn} : g_i(\mbp,\mbx) \leq 0, i \in [m] \rcb,
\end{align}
which can be seen as the Cartesian product of the safe-reachable region of $A_1$ and the time-extended safe-reachable regions of $A_2,\dots,A_m$. Let us also define the (global) optimal value function of \eqref{eq:main_prob} (our candidate for the Value function of the game) as
\begin{align} \label{eq:optval}
    V(\mbx) = \min_{\mbp} \lcb f_\theta (\mbp) : \mbp \in \cA(\mbx) \rcb,
\end{align}
and the solution set (i.e., the set of global minimizers) of the same problem as
\begin{align} \label{eq:solset}
    \cS(\mbx) = \lcb \mbp \in \cA(\mbx) : f_\theta (\mbp) = V(\mbx) \rcb.
\end{align}
Lastly, we define the set of (initial) states for which there is no single optimal capture point lying within $\cT$ as follows:
\begin{align} \label{eq:capreg}
    \cX = \Big\{ \mbx : h(p_i) > 0, \fa \mbp \in \cS(\mbx), i \in [m] \Big\}.
\end{align}

Next, we provide a few remarks summarizing some key properties (with sketches of proofs) of problem \eqref{eq:main_prob}.

\begin{remark} \label{rem:sol}
    The solution set $\cS(\mbx)$ is always nonempty and compact. To see this, note first that the feasible set $\cA(\mbx)$ is compact. This can be proven by showing, for instance, that if $\|\mbp\| \rightarrow \infty$, then there exists at least one $i \in [m]$ such that $g_i(\mbp,\mbx) \rightarrow \infty$, i.e., at least one attacker has a bounded (time-extended) safe-reachable region. Since, moreover, the proximity function $h$ is assumed to be continuous, the desired result follows directly from Weierstrass' Theorem \cite[Proposition~A.8]{bertsekas2016nonlinear}. 
\end{remark}

\begin{remark} \label{rem:ncvx}
    For multi-attacker scenarios (i.e., $m \geq 2$), the feasible set $\cA(\mbx)$ is not guaranteed to be convex. To see why, it suffices to note that the constraint functions $g_i$ (except for $i=1$) are concave in $p_j$ for all $j \in [i-1]$ (i.e., problem \eqref{eq:main_prob} is a difference-of-convex program \cite{shen_2016}). This alludes that computing and certifying $\cS(\mbx)$ can be challenging. The stationary points of the same problem can, however, be found efficiently via the convex-concave procedure \cite{shen_2016}.
\end{remark}

\begin{remark} \label{rem:third}
    If $\mbx \in \cX$, any local minimizer of problem \eqref{eq:main_prob} is attained on the boundary of $\cA(\mbx)$, and at such a point, the gradients $\nabla_{\mbp} g_i$ are well-defined (i.e., $p_1 \neq \dots \neq p_m$ and $p_i \neq x_{A_i}$ for all $i \in [m]$) and linearly independent. In addition, removing any constraint will strictly decrease the (global) optimal value of the problem.
\end{remark}

\subsection{Proof of Equilibrium}

In the remainder of this section, we focus on proving the optimality of the candidate equilibrium strategies described by \eqref{eq:eqstr_multi}. This task is equivalent to proving that the (global) optimal value computed by \eqref{eq:optval} equals the Value of the game, i.e., satisfies the related Isaacs equation which is to be provided below.

The Hamiltonian of the game can be defined (over the regular region) as
\begin{align} \label{eq:hamiltonian}
    H(\nabla U,u,v) = (\nabla_{x_D} U)^\top v + \sum\nolimits_{i=1}^m (\nabla_{x_{A_i}} U)^\top u_i,
\end{align}
where $U : \bR^{n(m+1)} \rightarrow \bR$ is the Value function of (the first phase of) the game, and $u = [u_1^\top,\dots,u_m^\top]^\top$. The Isaacs equation associated with the considered game can be written in the form \cite{bacsar1998dynamic}
\begin{subequations} \label{eq:HJI}
    \begin{align}
        \bar{H}(\nabla U) &= 0, & \mbx &\in \cX \backslash \cC_1, \label{eq:hji_main}
        \\
        U(\mbx) &= \Phi(\mbx), & \mbx &\in \partial\cC_1, \label{eq:hji_boundary}
    \end{align}
\end{subequations}
where $\bar{H} = \min_{u} \max_{v} H = \max_{v} \min_{u} H$ is the Hamiltonian evaluated along the equilibrium trajectory, and $\partial \cC_1$ is the boundary of $\cC_1$. An expression for this function can be readily obtained as
\begin{align}
    \bar{H}(\nabla U) = \len \nabla_{x_D} U \ren - \sum\nolimits_{i=1}^m \nu_i \len \nabla_{x_{A_i}} U \ren, \label{eq:opt_hamiltonian}
\end{align}
since it is immediate from \eqref{eq:hamiltonian} that the optimal control inputs of the agents are given by, for all $i \in [m]$,
\begin{align} \label{eq:eqstr_h}
    u_i^*(t) = -\frac{\nabla_{x_{A_i}} U}{\|\nabla_{x_{A_i}} U\|}, \quad v^*(t) = \frac{\nabla_{x_D} U}{\|\nabla_{x_D} U\|}.
\end{align}

To show that $V$ satisfies \eqref{eq:HJI}, it is best to have its gradient in closed form. Unfortunately, this is not always possible since problem \eqref{eq:main_prob} generally has no analytical solution. To address this challenge, we employ a classical (second-order) result from the sensitivity analysis of nonlinear programming \cite[Theorem~3.2.2,~Theorem~3.4.1]{fiacco}. Before introducing this result, a few related definitions are provided. The index set of active constraints at $(\mbp,\mbx)$ is defined as
\begin{align}
    \cI(\mbp,\mbx) = \lcb i \in [m] : g_i(\mbp,\mbx) = 0 \rcb.
\end{align}
The Lagrangian associated with problem \eqref{eq:main_prob} is defined as
\begin{align}
    &L(\mbp,\lambda,\mbx) = f_\theta(\mbp) + \sum\nolimits_{i=1}^m \lambda_i g_i(\mbp,\mbx),
\end{align}
where $\lambda = [\lambda_1,\dots,\lambda_m]^\top \in \bR^m$ is a Lagrange multiplier. The Karush-Kuhn-Tucker (KKT) conditions are \cite{bertsekas2016nonlinear}
\begin{subequations} \label{eq:kkt}
    \begin{align}
        \nabla_{(\mbp,\lambda)} L(\mbp,\lambda,\mbx) &= \zv, \label{eq:kkt_station}
        \\
        \fa i \in [m], ~ \lambda_i &\geq 0, \label{eq:kkt_dualfeas}
        \\
        \sum\nolimits_{i=1}^m \lambda_i g_i(\mbp,\lambda,\mbx) &= 0. \label{eq:kkt_compslack}
    \end{align}
\end{subequations}

\begin{theorem}[\cite{fiacco}] \label{theo:fiacco}
    Let $\bar\mbx \in \cX$. Let $\mbp^*$ be a local minimizer of problem \eqref{eq:main_prob} at $\bar\mbx$. If there exists a multiplier $\lambda^* $ such that the triple $(\mbp^*,\lambda^*,\bar\mbx)$ satisfies all of the conditions below:
    \begin{enumerate}
        \item the functions $f_\theta$ and $g_i$ (for all $i \in [m]$) are twice continuously differentiable in $(\mbp,\mbx)$ near $(\mbp^*,\bar\mbx)$;
        \item Linear Independence Constraint Qualification (LICQ) holds, i.e.,
        \begin{align} \label{eq:LICQ}
            &\nabla_\mbp g_i(\mbp^*,\bar\mbx), ~\fa i \in \cI(\mbp^*,\bar\mbx), \nonumber
            \\
            &\qquad\qquad \textrm{~are~linearly~independent;} \tag{LICQ}
        \end{align}
        \item Strict Complementarity Slackness (SCS) holds, i.e.,
        \begin{align} \label{eq:SCS}
            \lambda_i^* > 0, \quad \fa i \in \cI(\mbp^*,\bar\mbx); \tag{SCS}
        \end{align}
        \item Second-Order Sufficient Condition (SOSC) holds, i.e., $(\mbp^*,\lambda^*)$ satisfies the KKT conditions \eqref{eq:kkt} and
        \begin{align} \label{eq:SOSC}
            y^\top \nabla_{\mbp\mbp}^2 L(\mbp^*,\lambda^*,\bar\mbx) y > 0, ~ \fa y \in \cF(\mbp^*,\bar\mbx), \tag{SOSC}
        \end{align}
        where
        \begin{align*}
            \cF(\mbp,\mbx) = \big\{ y \neq \zv : \nabla_\mbp g_i(\mbp,\mbx)^\top y = 0, i \in \cI(\mbp,\mbx) \big\},
        \end{align*}
    \end{enumerate}
    then there exists a neighborhood $\cB(\bar\mbx,\epsilon) = \{\mbx : \|\mbx-\bar\mbx\|<\epsilon\}$, with $\epsilon > 0$, and continuously differentiable functions $\mbp^\star : \cB(\bar\mbx,\epsilon) \rightarrow \bR^{mn}$ and $\lambda^\star : \cB(\bar\mbx,\epsilon) \rightarrow [0,\infty)^m$ such that at any $\mbx \in \cB(\bar\mbx,\epsilon)$,
    \begin{enumerate}
        \item $\mbp^\star(\mbx) = \mbp^*$ and $\lambda^\star(\mbx) = \lambda^*$, and $\mbp^*$ is a strict local minimizer of problem \eqref{eq:main_prob} at $\bar\mbx$;
        \item the (local) optimal value function $W : \cB(\bar\mbx,\epsilon) \rightarrow \bR$ defined by
        \begin{align} \label{eq:local_optval}
            W(\mbx) = f_\theta(\mbp^\star(\mbx),\mbx)
        \end{align}
        is continuously differentiable and has the gradient
        \begin{align}
            \nabla W(\mbx) &= \nabla_\mbx L(\mbp^*,\lambda^*,\mbx).
        \end{align}
    \end{enumerate}
\end{theorem}

From Remark \ref{rem:third}, we know that the first three conditions in Theorem \ref{theo:fiacco} automatically hold at any local minimizer of problem \eqref{eq:main_prob}. Thus, it is only \eqref{eq:SOSC} that needs to be satisfied for us to make use of this theorem.

We are now ready to state our main theoretical result, which (partially) answers Problem \ref{prob:main}.

\begin{theorem} \label{theo:multi}
    At any $\mbx \in \cX$ where the set $\cS(\mbx)$ is a singleton,
    \begin{enumerate}
        \item the function $V$ defined in \eqref{eq:optval} is continuously differentiable and satisfies \eqref{eq:HJI};
        \item the control policies described by \eqref{eq:eqstr_multi} are the corresponding equilibrium strategies of the game.
    \end{enumerate}
\end{theorem}

\begin{proof}
    If $\cS(\mbx)$ is a singleton $\{\mbp^\star\}$, the four conditions in Theorem \ref{theo:fiacco} are satisfied. That is, $V$ is continuously differentiable, and its gradient can be expressed as
    \begin{align} 
        \nabla V(\mbx) &= \nabla_\mbx L(\mbp^\star,\lambda^\star,\mbx) \nonumber
        \\
        &= \underbrace{\nabla_\mbx f_\theta(\mbp^\star)}_{= \, \zv} + \sum\nolimits_{i=1}^m \lambda_i^\star \nabla_\mbx g_i(\mbp^\star,\mbx) \nonumber
        \\
        &=
        \begin{bmatrix}
            -\dfrac{\lambda_1^\star}{\nu_1} \dfrac{p_1^\star - x_{A_1}}{\|p_1^\star - x_{A_1}\|}
            \\
            \vdots
            \\
            -\dfrac{\lambda_m^\star}{\nu_m} \dfrac{p_m^\star - x_{A_m}}{\|p_m^\star - x_{A_m}\|}
            \\
            \lp \sum_{i=1}^m \lambda_i^\star \rp \dfrac{p_1^\star - x_{D}}{\|p_1^\star - x_{D}\|}
        \end{bmatrix}, \label{eq:grad}
    \end{align}
    where $(\mbp^\star,\lambda^\star,\mbx)$ is the unique triple satisfying the KKT conditions in \eqref{eq:kkt}. Substituting \eqref{eq:grad} into \eqref{eq:opt_hamiltonian} readily yields
    \begin{align}
        \bar H(\nabla V(\mbx)) &= \len \lp \sum\nolimits_{i=1}^m \lambda_i^\star \rp \dfrac{p_1^\star - x_{D}}{\|p_1^\star - x_{D}\|} \ren \nonumber
        \\
        &\quad - \sum\nolimits_{i=1}^m \nu_i \len \lp -\frac{\lambda_i^\star}{\nu_i} \rp \frac{p_i^\star - x_{A_i}}{\|p_i^\star - x_{A_i}\|} \ren \nonumber
        \\
        &= \left| \sum\nolimits_{i=1}^m \lambda_i^\star \right| -\sum\nolimits_{i=1}^m \left| \lambda_i^\star \right| \overset{\eqref{eq:kkt_dualfeas}}{=} 0.
    \end{align}
    Hence \eqref{eq:hji_main} is satisfied. In addition, by substituting \eqref{eq:grad} in \eqref{eq:eqstr_h}, we can retrieve \eqref{eq:eqstr_multi}.

    Next, we check the boundary condition \eqref{eq:hji_boundary}. If $\mbx \in \partial\cC_1$, then $x_{A_1} = x_D$. In this case, the safe-reachable region of $A_1$ degenerates to a singleton $\{p_1^\star\}$, and problem \eqref{eq:main_prob} becomes
    \begin{equation}
        \begin{aligned}
            \min_{\mbp_{2:m}} ~~& \theta_1 h(p_1^\star) + \sum\nolimits_{i=2}^m \theta_i h(p_i)
            \\
            \mathrm{s.t.} ~~& \frac{\|p_i-x_{A_i}\|}{\nu_i} - \sum\nolimits_{j=2}^m \|p_j-p_{j-1}\| \leq 0, ~~ i \geq 2,
        \end{aligned} \nonumber
    \end{equation}
    with $p_1 = p_1^\star$, whose optimal value is by definition equivalent to $\Phi(\mbx)$ (see \eqref{eq:tercost} and \eqref{eq:prob_series}). This completes the proof.
\end{proof}

\begin{remark} \label{rem:sing}
    Problem \eqref{eq:main_prob} often admits more than one (but finitely many) strict global minimizers, in which case Theorem \ref{theo:multi} no longer holds. In such a case, $A$ is endowed with a freedom to choose among more than one equilibrium strategies, all leading to the same Value of the game, while $D$ faces a momentary dilemma as to which strategy $A$ will select. We thus infer that there exists a dispersal surface of attackers in this game, yet its identification is left for future work. Interested readers are referred to \cite{deng_sing} for an in-depth discussion on singular surfaces that exist in a similar class of games with two (homogeneous) attackers.
\end{remark}

\begin{remark}
    As mentioned in Remark \ref{rem:ncvx}, problem \eqref{eq:main_prob} may often have a nonconvex feasible region. In such a case, computing global minimizers or certifying the global optimality of a local minimizer can be challenging. Nevertheless, if a local minimizer $\mbp^*$ satisfies \eqref{eq:SOSC}, then it can be shown using Theorem \ref{theo:fiacco}, as similarly done in the proof of Theorem \ref{theo:multi}, that the associated local optimal value function $W$ defined in \eqref{eq:local_optval} satisfies \eqref{eq:hji_main} (but not necessarily \eqref{eq:hji_boundary}). That is, the value of $W$ will remain constant along the (open-loop) game trajectory generated with the corresponding (suboptimal) equilibrium control policies, i.e., \eqref{eq:eqstr_multi} with $\mbp^*$ replacing $\mbp^\star$.
\end{remark}

\section{Numerical Example} \label{sec:sim}

In this section, we present numerical examples to demonstrate how cooperative behaviors emerge within the attacker team as a consequence of their heterogeneity. The (fixed) parameters selected for the simulation results presented below are as follows: $n=2$, $m=6$, $x_{D}^0 = [2,2]^\top$, $\nu_1 = 1/5.5$, $\nu_2 = 1/6$, $\nu_3 = 1/6.5$, $\nu_4 = 1/7$, $\nu_5 = 1/7.5$, $\nu_6 = 1/8$, $\theta_1 = 1/22$, $\theta_2 = 5/22$, $\theta_3 = 1/22$, $\theta_4 = 6/22$, $\theta_5 = 1/22$, and $\theta_6 = 8/22$ (see also Figure \ref{fig:par}). The initial positions of the attackers are randomly generated such that $\|x_{A_i}^0\| = 10$. The target area is considered to be a closed disk with a radius $r=2$, i.e., $h(x) = \|x\|^2 - 4$.

To better highlight the emergence of cooperative attacker behaviors, we compare our results with those obtained by applying the solution presented in \cite{fu2021strategies} for noncooperative scenarios under identical simulation settings. In \cite{fu2021strategies}, a team of \emph{noncooperative} attackers solve a sequence of optimization problems (for all $i \in [m]$):
\begin{equation*}
    \begin{aligned}
        \min_{p_i} \quad& h(p_i)
        \\
        \mathrm{s.t.} \quad& \frac{\|p_i - x_{A_i}\|}{\nu_i} - \sum\nolimits_{j=1}^{i} \|p_j - p_{j-1}\| \leq 0,
    \end{aligned}
\end{equation*}
which is similar to problem \eqref{eq:prob_series} except that the decisions of attackers are hierarchical only through the constraints. That is, each attacker is indifferent to the objective values of the other attackers. Note also that the weight parameters $\theta_i$ are absent in this model.

The main simulation results are illustrated in Figures \ref{fig:sims} and \ref{fig:statistics}. In Figure \ref{fig:sims}, it can be observed that the attackers with low weight values (i.e., $A_1$, $A_3$, and $A_5$) naturally exhibit sacrificial behaviors to aid other attackers in approaching $\cT$ more closely. This can further be confirmed by the extensive statistical results portrayed in Figure \ref{fig:statistics}, where it is evident that the proximity levels of $A_1$, $A_3$, and $A_5$ (resp., $A_2$, $A_4$, and $A_6$) are noticeably higher (resp., lower) compared to the noncooperative scenario.

\section{Conclusion} \label{sec:concl}

In this paper, we have designed and verified the equilibrium control policies of a single defender and a heterogeneous group of cooperative attackers with conflicting objectives. In the future, we plan to explore a variation of this problem wherein the defender has no prior (or incorrect) information about the speed ratios and degrees of importance of the attackers.

\begin{figure}[h]
    \centering
    \subcaptionbox{Speed ratios}{
    \includegraphics[scale=0.44]{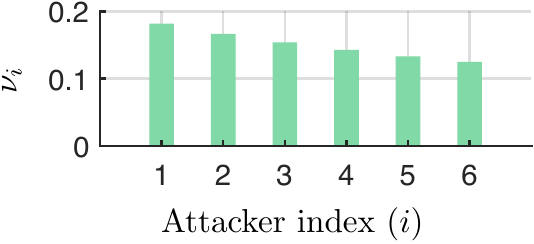}
    }
    \subcaptionbox{Weights}{
    \includegraphics[scale=0.44]{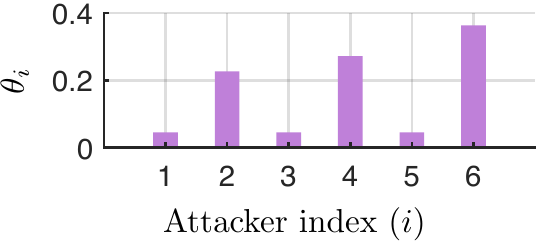}
    }
    \caption{Selected parameter values.}
    \label{fig:par}
\end{figure}

\begin{figure}[h]
    \centering
    \subcaptionbox{Noncooperative \cite{fu2021strategies} \label{fig:simsa}}{
    \includegraphics[scale=0.33]{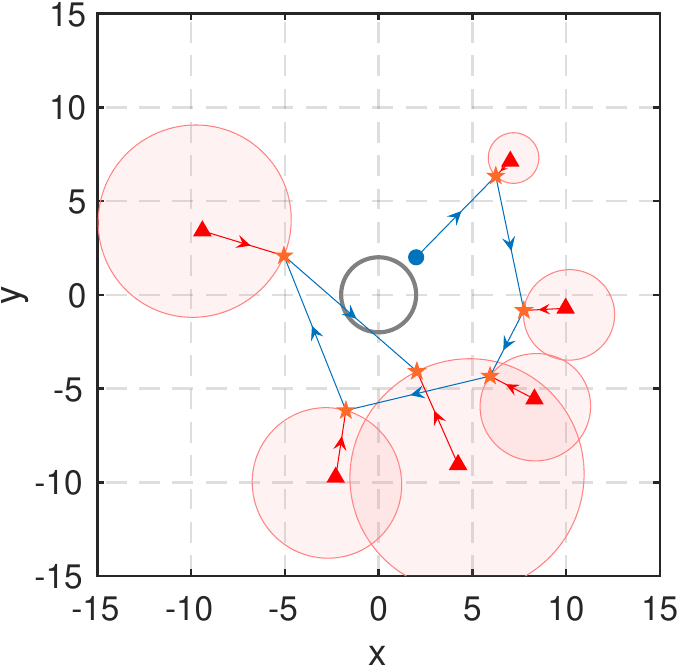}
    }
    \subcaptionbox{Cooperative (ours) \label{fig:simsb}}{
    \includegraphics[scale=0.33]{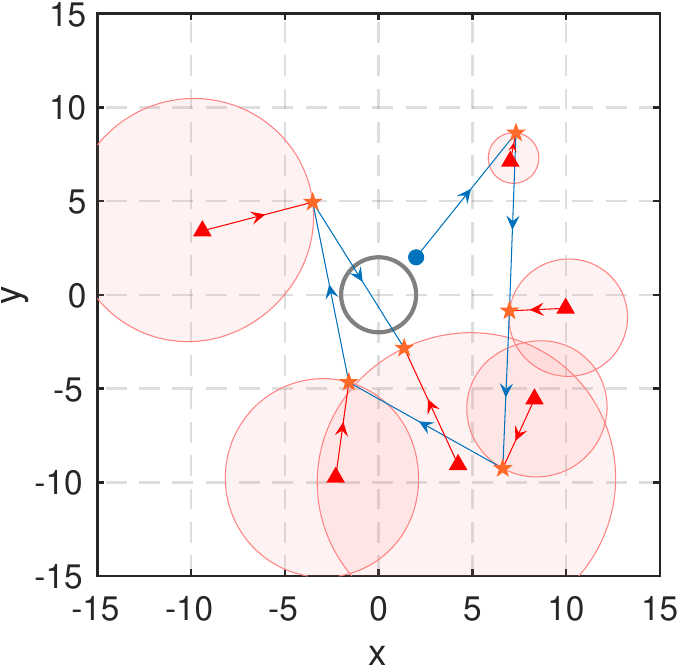}
    }
    \caption{Equilibrium trajectory of the game. The blue and the red markers indicate the defender and the attackers, respectively. The pink regions illustrate the (time-extended) safe-reachable regions of the attackers, whereas the grey circle illustrates the boundary of the target area.}
    \label{fig:sims}
\end{figure}

\begin{figure}[h]
    \centering
    \subcaptionbox{Noncooperative \cite{fu2021strategies} \label{fig:statisticsa}}{
    \includegraphics[scale=0.4]{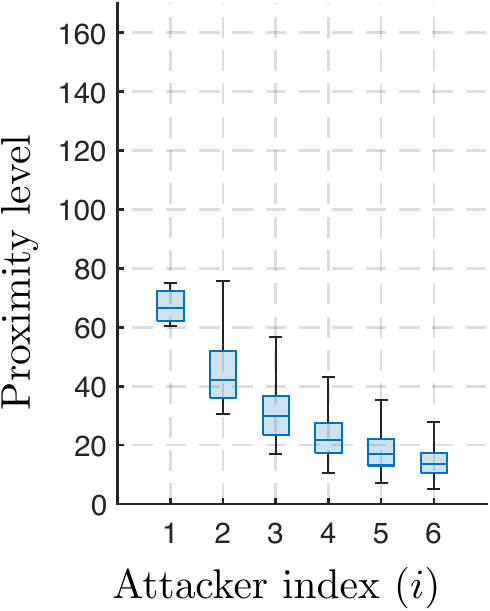}
    }
    \quad
    \subcaptionbox{Cooperative (ours) \label{fig:statisticsb}}{
    \includegraphics[scale=0.4]{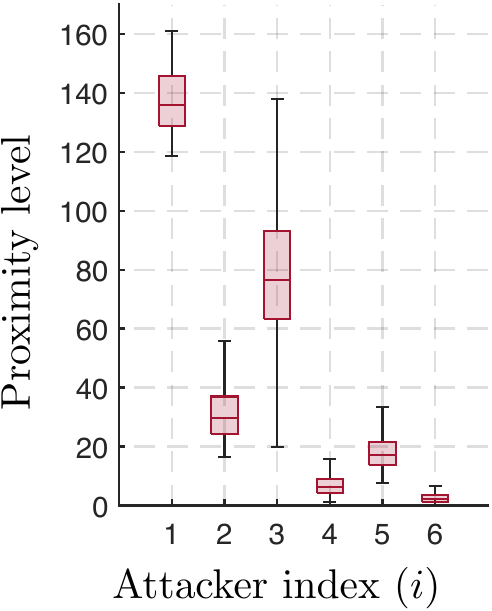}
    }
    \caption{Proximity levels of optimal capture points for $10^3$ randomly generated initial attacker positions.}
    \label{fig:statistics}
\end{figure}

\bibliographystyle{ieeetr}
\bibliography{root}

\end{document}

%% file: notation.tex
\usepackage{amsmath,amsfonts,amssymb,amsthm}
\usepackage{graphicx}
\usepackage{bm}
\usepackage[noadjust]{cite}
\usepackage{subcaption}

\newcommand{\fa}{\forall\,}

\newcommand{\lp}{\left(}
\newcommand{\rp}{\right)}
\newcommand{\len}{\left\|}
\newcommand{\ren}{\right\|}

\newcommand{\lcb}{\left\{}
\newcommand{\rcb}{\right\}}
\newcommand{\zv}{{\bm{0}}}

\newcommand{\mbx}{\mathbf{x}}

\newcommand{\mbp}{\mathbf{p}}

\newcommand{\cA}{\mathcal{A}}
\newcommand{\cB}{\mathcal{B}}
\newcommand{\cC}{\mathcal{C}}

\newcommand{\cF}{\mathcal{F}}

\newcommand{\cI}{\mathcal{I}}

\newcommand{\cS}{\mathcal{S}}
\newcommand{\cT}{\mathcal{T}}
\newcommand{\cU}{\mathcal{U}}

\newcommand{\cX}{\mathcal{X}}

\newcommand{\bR}{\mathbb{R}}

\theoremstyle{definition}
\newtheorem{remark}{Remark}

\newtheorem{problem}{Problem}
\newtheorem{theorem}{Theorem}

\usepackage{xcolor}